\documentclass[11pt]{article}
\usepackage{mathrsfs}
\usepackage{tikz}
\usetikzlibrary{calc}
\usetikzlibrary{decorations.markings,arrows}
\usepackage{subcaption}

\usepackage{amsmath,amsthm,verbatim,amssymb,amsfonts,amscd, graphicx, float}
\usepackage{graphics}
\usepackage{enumitem}
\usepackage{bbm}
\usepackage[margin=1.1in]{geometry}
\usepackage{thmtools}
\usepackage{thm-restate}
\usepackage{hyperref}
\usepackage[noabbrev]{cleveref}
\Crefname{subsection}{subsection}{subsections}

\newtheorem{theorem}{Theorem}[section]

\newtheorem{lemma}[theorem]{Lemma}
\newtheorem{conjecture}[theorem]{Conjecture}
\newtheorem{claim}[theorem]{Claim}
\newtheorem{proposition}[theorem]{Proposition}
\newtheorem{remark}[theorem]{Remark}
\newtheorem{definition}[theorem]{Definition}

\newtheorem*{C1}{Claim~\ref{clm: J1 lower bound}}
\newtheorem*{C2}{Claim~\ref{clm: J0 lower bound}}
\newtheorem*{C3}{Claim~\ref{clm: J1r lower bound}}
\newtheorem*{C4}{Claim~\ref{clm: J0r lower bound}}

\newcommand\restr[2]{{
  \left.\kern-\nulldelimiterspace 
  #1 
  \vphantom{\big|} 
  \right|_{#2} 
  }}

\newcommand{\weight}{\mathrm{wt}}
\newcommand{\tweight}{\widetilde{\mathrm{wt}}}
\newcommand{\rweight}{\mathrm{wt}_r}

\newcommand\given[1][]{:}
\renewcommand{\Pr}{\mathop{\bf Pr\/}}

\newcommand{\indicator}{\mathbbm{1}}

\newcommand{\CWE}{C^{\alpha}_{\text{WE}}}
\newcommand{\CWEr}{C^{\alpha,r}_{\text{WE}}}
\newcommand{\CWEs}{C^{\alpha^*}_{\text{WE}}}
\newcommand{\CWErs}{C^{\alpha^*,r}_{\text{WE}}}

\newcommand{\field}{\mathbb{F}}
\newcommand{\F}{\mathbb{F}}

\newcommand{\mmod}[1]{\ \mathrm{mod}\ #1}

\definecolor{MyDarkBlue}{rgb}{0,0.08,1}
\definecolor{MyDarkGreen}{rgb}{0.02,0.6,0.02}
\definecolor{MyDarkRed}{rgb}{0.8,0.02,0.02}
\definecolor{MyDarkOrange}{rgb}{0.40,0.2,0.02}
\definecolor{MyPurple}{RGB}{111,0,255}
\definecolor{MyRed}{rgb}{1.0,0.0,0.0}
\definecolor{MyGold}{rgb}{0.75,0.6,0.12}
\definecolor{MyDarkgray}{rgb}{0.66, 0.66, 0.66}

\date{}

\begin{document}
\title{\textbf{A Deterministic Construction of a Large Distance Code from the Wozencraft Ensemble\footnote{An extended abstract of this manuscript was presented at RANDOM 2023.}}}
\author{
Venkatesan Guruswami\thanks{{Departments of EECS and Mathematics, and the Simons Institute for the Theory of Computing, UC Berkeley, Berkeley, CA, 94709, USA. Email: \url{venkatg@berkeley.edu}. Research supported by a Simons Investigator Award and NSF grants CCF-2210823 and CCF-2228287.}}
\and 
Shilun Li\thanks{Department of Mathematics, UC Berkeley, Berkeley, CA, 94709, USA. Email: \url{shilun@berkeley.edu}. Research supported by University of California, Berkeley under Berkeley Fellowship.} 
}
\maketitle
\begin{abstract}
We present an explicit construction of a sequence of rate $1/2$ Wozencraft ensemble codes (over any fixed prime field $\F_q$) that achieve minimum distance $\Omega(\sqrt{k})$ where $k$ is the message length. The coefficients of the Wozencraft ensemble codes are constructed using Sidon Sets and the cyclic structure of $\field_{q^{k}}$ where $k+1$ is prime with $q$ a primitive root modulo $k+1$. Assuming Artin's conjecture, there are infinitely many such $k$ for any prime $q$. 
\end{abstract}

\section{Introduction}
The explicit construction of binary error-correcting codes with a rate vs. distance trade-off approaching that of random constructions, i.e., the so-called Gilbert-Varshamov (GV) bound, remains an outstanding challenge in coding theory and combinatorics.

For large $n$, a random binary linear code of rate $R \in (0,1)$, defined for example as the column span of a random matrix $G \in \F_2^{n \times Rn}$, has relative distance $h^{-1}(1-R)$ with high probability, where $h^{-1}(\cdot)$ is the inverse of the binary entropy function. There is a similar GV bound $h_q^{-1}(1-R)$, involving the $q$-ary entropy function, for codes over other finite fields $\F_q$.

While explicit constructions meeting the GV bound remain elusive\footnote{Ta-Shma \cite{ta2017explicit} constructed explicit binary codes near the GV bound for low rates. The codes have distance $\frac{1-\epsilon}{2}$ and rate $\Omega(\epsilon^{2+o(1)})$ which is asymptotically close to the rate $\Omega(\epsilon^2)$ guaranteed by the GV bound.}, there are known derandomizations showing that codes drawn randomly from much smaller, structured ensembles can also achieve the GV bound. One of the most classical and famous such ensemble is the Wozencraft ensemble, which consists of codes $\CWE=\{(x,\alpha x):x\in\field_{q^{k}}\}$ as $\alpha$ varies over nonzero elements of the field $\F_{q
^k}$, and one uses some fixed basis to express elements of $\F_{q^k}$ as length $k$ vectors over $\F_q$. Note that each code $\CWE$ has rate $1/2$. The construction of Wozencraft ensemble codes $\CWE$ first appeared in a paper by Massey \cite{massey1963threshold}, who attributed the discovery of these codes to John M. Wozencraft.

It is a standard exercise to show that for most choices of $\alpha$, the code $\CWE$ has distance close to $h_q^{-1}(1/2)$ and thus achieves the GV bound.  Puncturing the Wozencraft ensemble gives codes of higher rates that also meet the GV bound. The property of the Wozencraft ensemble and its punctured variants behind this phenomenon is that every nonzero word appears as a codeword in an equal number of codes in the ensemble (if it appears in any code of the ensemble at all).

Using this property, Justesen~\cite{justesen1972class} in 1972 gave the first strongly explicit asymptotically good binary linear codes, by concatenating an outer Reed–Solomon code with different inner codes drawn from the Wozencraft ensemble for different positions of the Reed-Solomon code. The Justesen construction achieves a trade-off between rate vs. distance called the Zyablov bound for rates at least $0.3$. There are variants of Wozencraft codes which give Zyablov-bound-achieving codes for lower rates as well (see Section~\ref{sec:open}).
In recent years, Wozencraft ensemble codes have also found other varied uses, for example in constructing covering codes of small block length~\cite{potukuchi2020improved} and  write-once-memory codes~\cite{shpilka2012capacity, shpilka2013new}.

Given that for most $\alpha$, the code $\CWE$ meets the GV bound, it is a natural question whether one can find an explicit $\alpha$ for which the code has good distance (even if it doesn't quite meet the GV bound). 
Gaborit and Zemor \cite{gaborit2008asymptotic} showed that it suffices to consider random $\alpha$ in a subset of size $\approx q^k/k$ and used it to show the existence of linear codes which are a factor $k$ larger in terms of size than the Gilbert-Varshamov bound (such a result was shown earlier for general codes in \cite{JiangV04}).

However, it remains an outstanding challenge to find some $\alpha$ in deterministic $\text{poly}(k)$ time for which $\CWE$ has distance $\Omega(k)$. This question is relatively well-known, eg. it received mention in a blog post by Lipton~\cite{lipton_2011}, but has resisted progress. To the best of our knowledge, even an explicit $\alpha$ for which $\CWE$ has distance $k^{\Omega(1)}$ was not known. 

For certain structured fields $\F_{q^k}$ (of which there are an infinite family under Artin's conjecture), we give an explicit construction of $\alpha \in \F_{q^k}$ for which $\CWE$ has distance $\Omega(\sqrt{k})$. We also give an explicit puncturing of these codes to achieve any desired rate $r < 1$, and $\Omega_r(\sqrt{k})$ distance (the constant in the $\Omega()$ depends on $r$). Our theorems are informally stated below.
\begin{theorem}[Informal]
\label{thm: informal}
Fix a prime field $\F_q$ and consider an integer
$k$ such that $k+1$ is prime and $q$ is a primitive root modulo $k+1$. There exist $\alpha^*\in \field_{q^k}$ which can be constructed in deterministic $\text{poly}(k)$ time such that:
\begin{itemize}
    \item $\CWEs$ has distance $\Omega(\sqrt{k})$.
    \item For any $r>\frac{1}{2}$, there is an explicit puncturing of $\CWEs$ with rate at least $r$ and distance $\Omega_r(\sqrt{k}).$
\end{itemize}
\end{theorem}
Please refer to \Cref{thm: rate 1/2 construction} and \Cref{thm: higher rate construction} for construction of $\alpha^*$ and choice of puncturing.




\section{Preliminaries}
\label{section notation}
Throughout this paper, we will assume the alphabet has size $q$, where $q$ is prime. Furthermore, we will assume $k'$ is a prime such that $q$ is a primitive root modulo $k'$. A widely believed conjecture by Artin stats the following \cite{heath1986artin}:
\begin{conjecture}[Artin's Conjecture on Primitive Roots]
For any integer $m$ that is neither a square number nor $-1$, it is a primitive root modulo infinitely many primes $p$. Moreover, the set of prime numbers $p$ such that $m$ is a primitive root modulo $p$ has a positive asymptotic density inside the set of all primes.
\end{conjecture}
Assuming Artin's conjecture, such $k'$ exists infinitely often at sufficiently high density and can be efficiently found in deterministic $\mathrm{poly}(k')$ time.

\begin{remark}
    Heath proved a weaker density version of Artin's conjecture that holds for all but at most two primes \cite{heath1986artin}. This result directly leads to an unconditional conclusion: for all but at most two choices of the prime alphabet size $q$, primes $k'$ such that $q$ is a primitive root modulo $k'$ exists infinitely often at sufficiently high density and can be efficiently found in deterministic $\mathrm{poly}(k')$ time.
\end{remark}

Denote $k = k'-1$ for ease of notation. Let $p(x) = 1+x+x^2 + \ldots + x^{k'-1}$ be the cyclotomic polynomial which is irreducible over $\F_q$ (see \Cref{prop:irred-poly}). Then $\F_q[x]/(p(x))\cong \F_{q^k}$ for the extension field $\F_q[x]/(p(x))$ and we will fix the representation of $\F_{q^k}$ as polynomials in $x$ of degree less than $k$, with operations performed modulo $p(x)$. We will fix the $\F_q$-linear isomorphism $\varphi : \F_{q^k} \to \F_q^k$ that maps a polynomial of degree less than $k$ to its coefficient vector.  That is, $\varphi(\sum_{i=0}^{k-1} a_i x^i) = (a_0,a_1,\dots,a_{k-1})$. For $\alpha \in \F_{q^k}$, define $\weight(\alpha)$ to be the Hamming weight of $\varphi(\alpha)$. 

The Wozencraft ensemble is a classic family of codes defined as follows.

\begin{definition}

For $\alpha \in \F_{q^k}$, the \textbf{Wozencraft ensemble code} $\CWE$ parameterized by $\alpha$ is given by
$$\CWE=\{(\varphi(x),\varphi(\alpha x)):x\in \field_{q^{k}}\} \ . $$
Note that $\CWE$ is an $\F_q$-linear code of rate $1/2$.
\end{definition}

\begin{proposition}
\label{prop:irred-poly}
    If $k'$ is a prime such that $q$ is a primitive root modulo $k'$, then $p(x)=\sum_{i=0}^{k'-1}x^i$ is an irreducible polynomial of degree $k'-1$ in $\field_q[x]$.
\end{proposition}
\begin{proof}
    Since $q$ is a primitive root modulo $k'$, $k'-1$ is the smallest integer $d$ satisfying $q^d\equiv 1\mmod{k'}$. Note that a field extension $\field_{q^d}$ of $\field_q$ contains a primitive $k'$-th root of unity $\zeta$ if and only if $k'|q^d-1$, i.e., $q^d\equiv 1\mmod{k'}$. So $\field_{q^{k'-1}}$ is the smallest field extension of $\field_q$ containing $\zeta$ and the mimimal polynomial of $\zeta$ has degree $k'-1$. Since $p=\sum_{i=0}^{k'-1}x^i$ is a degree $k'-1$ polynomial such that $p(\zeta)=0$, $p$ is the minimal polynomial of $\zeta$ and thus irreducible. 
\end{proof}

We will use Sidon sets \cite{sidon1932satz, babcock1953intermodulation} to construct the parameter $\alpha$ for $\CWE$.
\begin{definition}
A \textbf{Sidon set} is a set of integers $A=\{a_1,\ldots,a_d\}$ where $a_1<a_2<\ldots<a_d$ such that for all $i,j,k,l\in [d]$ with $i\neq j$ and $k\neq l$,
$$a_i-a_j=a_k-a_l \Longleftrightarrow i=k\text{ and }j=l.$$
A \textbf{Sidon set modulo n} is a \textbf{Sidon set} such that for all $i,j,k,l\in [d]$ with $i\neq j$ and $k\neq l$,
$$a_i-a_j\equiv a_k-a_l \pmod{n} \Longleftrightarrow i=k\text{ and }j=l.$$
Size $d$ of the Sidon set $A$ is referred to as its \textbf{order} and $a_d-a_1$ as its \textbf{length}.
\end{definition}

\begin{remark}
For any Sidon set with order $d$ and length $m$, the $\binom{d}{2}$ distances between each pair of points need to be distinct. So $m\geq \binom{d}{2}$ and this gives a trivial upper bound $d\leq \sqrt{2m}$. Erd\H{o}s and Tur\'{a}n proved an upper bound of $d\leq \sqrt{m}+O(m^{1/4})$ \cite{erdos1941problem}, with an alternative proof by Lindstr{\"o}m \cite{lindstrom1972b2}. Both proofs in fact give a sharper upper bound of $\sqrt{m}+m^{1/4}+1$ by carefully examining the inequalities in their proofs \cite{obryant2022sizefinitesidonsets}. On the other hand, it is believed that the maximal $d$ given $m$ satisfies $d>\sqrt{m}$ \cite{dimitromanolakis2002analysis}. Erd\H{o}s and Tur\'{a}n\cite{erdos1941problem} gave explicit constructions where $d\geq \sqrt{m}-O(m^{5/16})$.
\end{remark}
We will introduce the Bose-Chowla construction of Sidon sets \cite{bose1960theorems, JIMSIMS151305}.
\begin{theorem}[Bose-Chowla, \cite{bose1960theorems}]
    Let $q$ be a power of a prime, $\theta$ be a primitive root in $\field_{q^2}$. Then the sequence of $q$ integers
    $$A=\{a:1\leq a < q^2\text{ and }\theta^a - \theta\in \field_{q^2}\cap \field_{q}\}$$
    forms a Sidon set modulo $q^2-1$.
\end{theorem}
This construction of Sidon set has order $d=q$ and length at most $m(d)=q^2-2$. They are asymptotically optimal in the sense that $\lim_{d\rightarrow \infty}\sqrt{m(d)}/d=1$. Given $q$, such construction can be done in $O(q^4)$ time by finding a primitive root $\theta$ via naive search and checking the inclusion $\theta^a - \theta\in \field_{q^2}\cap \field_{q}$ for every $a$. We will later define our parameter $\alpha$ to be $\sum_{a\in A}x^a$ where $A$ is a Sidon set and show that $\CWE$ has good distance.

\section{Rate 1/2 Construction}
\label{section 1/2}
In this section, we will give explicit construction of rate $1/2$ Wozencraft ensemble codes with minimum distance $\Omega(\sqrt{k})$ using Sidon sets. To begin, we provide an intuitive explanation for the natural occurrence of Sidon sets in this specific context. Subsequently, we proceed with the analysis of the minimum distance of our construction.

\subsection{Motivation}
Fix a set of indices $A\subseteq [k]$ and an element $\alpha=\sum_{a\in A}x^a\in \field_{q^k}$ with coefficients either $0$ or $1$. Take any $y=\sum_{s\in S}b_s x^s$ where $S$ is the set of non-zero indices of the coefficients of $y$, so $b_s\neq 0$ for all $s\in S$. To establish a lower bound on $\Delta(\CWE)$, where $\Delta(\CWE)$ denotes the distance of the code $\CWE$, we would like show a lower bound on the weight of the product $\alpha y\in \field_{q^k}$ for any $y$. To simplify the analysis, we consider a ring extension of $\field_{q^k}$ (described in \cref{sec: proof of rate 1/2}), in which the coefficient $c_j$ in front of $x^j$ of the product $\alpha y$ can be expressed as $c_j=
\sum_{a\in A,s\in S}\indicator\{a+s\equiv j\pmod{k'}\}b_s$. We would like to establish a lower bound on the number of non-zero coefficients $c_j\neq 0$, which would transform into a lower bound on the weight of $\alpha y$ in $\field_{q^k}$. It is sufficient to demonstrate the existence of numerous choices of $j$ satisfying $j\equiv a+s \pmod{k'}$ for a unique combination of $a$ and $s$. In this case, $c_j$ corresponds to the sum of a single non-zero element and is therefore non-zero. To ensure there are an abundance of such choices for $j$ with this uniqueness property, it is desirable to minimize collisions of the form $a+s\equiv a'+s'\pmod{k'}$ where $a,a'\in A$ and $s,s'\in S$. Since $S$ can be selected adversarially with respect to $A$, it is advantageous to have $(a-a')\mod{k'}$ be unique, which is exactly the property of Sidon sets modulo $k'$. The result of this construction is stated formally as follows, with the proof provided in \cref{sec: proof of rate 1/2}:
\begin{theorem}
\label{thm: rate 1/2 construction}
    Fix a prime field $\F_q$ and consider an integer
$k$ such that $k+1$ is prime and $q$ is a primitive root modulo $k+1$. Let $d$ be the largest prime smaller than $\sqrt{k}$ and let $A=\{a_1,\ldots,a_d\}$ be a Bose-Chowla Sidon set with order $d$. Define $\alpha^*=\sum_{a\in A} x^{a}\in\field_{q^k}$. Then $\Delta(\CWEs)\geq \frac{d}{2}=\Omega(\sqrt{k})$.
\end{theorem}
\begin{remark}
\label{remark d}
Since there exists a prime between $[\frac{1}{2}\sqrt{k},\sqrt{k}]$ by Bertrand–Chebyshev theorem \cite{vcebyvsev1850memoire}, $d$ can be found efficiently via naive search and $d\geq \frac{1}{2}\sqrt{k}$. Moreover, Baker, Harman and Pintz \cite{baker2001difference} showed that there exists a prime in the interval $[\sqrt{k}-k^{0.27}, \sqrt{k},]$ for $k$ sufficiently large. So $d=(1-o(1))\sqrt{k}$ and the constructed code $\CWEs$ has distance asymptotically $\Delta(\CWEs)\geq (1-o(1))\sqrt{k}$ as $k\rightarrow\infty$.
\end{remark}

It is also worthwhile to note that when $a+s\equiv j\pmod{k'}$ holds for more than one pair of $(a,s)$, $c_j$ may still be non-zero as it is a sum of multiple non-zero elements. In fact, for $\alpha$ with large weight, it is common that $\alpha y$ has substantial weight yet few choices of $j$ satisfy the uniqueness property. One possible approach to improving the construction of $A$ involves analyzing scenarios where $j\equiv a+s \pmod{k'}$ for multiple pairs of $(a,s)$.

\subsection{Proof of \Cref{thm: rate 1/2 construction}}
\label{sec: proof of rate 1/2}

To analyze the minimum distance of $\CWE$, it is helpful to define the ring $R=\field_q[x]/(x^{k'}-1)$, which consists of polynomials of degree less than $k'=k+1$. Recall that $p(x)=1+x^2+\cdots+x^{k'-1}$ divides $x^{k'}-1$, therefore we can identify $\field_{q^k}\cong R/(p)$ by the map sending $f\in R$ to $(f\mmod{p})\in\field_{q^k}$.
In addition, we can consider the $\field_q$-vector space embedding $\field_{q^{k}}\subseteq R$ and extend $\varphi$ to the $\field_q$-linear map $\widetilde{\varphi}: R\rightarrow \field_q^{k'}$ mapping polynomials of degree less than $k'$ to its coefficient vector. Define $\tweight(f)$ to be the Hamming weight of $\widetilde{\varphi}(f)$ for any $f\in R$. The following lemma gives the relationship between $\tweight(f)$ and $\weight(f\mmod p)$.

\begin{lemma}
    \label{lem: relation of wt}
    For any $f\in R$, $\weight(f\mmod p)\geq \min\{\tweight(f),k-\tweight(f)\}$.
\end{lemma}
\begin{proof}
    For any $f\in R$, let us write $f=\sum_{i=0}^k b_i x^i$. Then 
    $$f\mmod p=f-b_k p=\sum_{i=0}^{k-1}(b_i-b_k)x^i.$$
    If $b_k=0$, then $\weight(f\mmod p)=\tweight(f)$; if $b_k\neq 0$, then 
    $$\weight(f\mmod p)=\left|\{i:b_i\neq b_k\}\right|\geq k-\tweight(f).$$
    So $\weight(f\mmod p)\geq \min\{\tweight(f),k-\tweight(f)\}$.
\end{proof}

\begin{remark}
    \Cref{lem: relation of wt} relies crucially on our choice of $p(x)=1+x^2+\cdots+x^{k'-1}$.
\end{remark}

\begin{lemma}
    \label{lem: equivalent condition}
    Given $\alpha\in \field_{q^k}\subseteq R$, suppose that for every $y\in R$ with $\tweight(y)\leq c(k)$ the condition
    $$c(k)-\tweight(y)\leq \tweight(\alpha y)\leq k-\left(c(k)-\tweight(y)\right)$$
    holds, where the product $\alpha y$ is taken in $R$. Then  $\Delta(\CWE)\geq c(k)$.
\end{lemma}
\begin{proof}
Take any non-zero $y\in \field_{q^k}\subseteq R$. Its corresponding codeword $\CWE(y)=(\varphi(y),\varphi(\alpha y\mmod{p}))$ has Hamming weight $\weight(y)+\weight(\alpha y \mmod{p})$. By \Cref{lem: relation of wt}, the above condition implies
$$\weight(y)+\weight(\alpha y\mmod{p})\geq \tweight(y)+\min\{\tweight(\alpha y),k-\tweight(\alpha y)\}\geq c(k).$$
Since $\CWE$ is a linear code, $\Delta(\CWE)\geq c(k)$.
\end{proof}

We can now prove \cref{thm: rate 1/2 construction} by showing that $\alpha^*$ satisfies the condition of \Cref{lem: equivalent condition} with $c(k)=d$.
\begin{proof}(\cref{thm: rate 1/2 construction})
Note that $\alpha^*$ is an element of $\field_{q^k}$ since it has degree at most $k-2$ by construction. Let us check that the condition of \Cref{lem: equivalent condition} indeed holds for $\alpha^*=\sum_{a\in A}x^a$ and $c(k)=\frac{d}{2}$. For any $y\in R$ with $\tweight(y)=w$, we can write $y=\sum_{i=1}^w b_{s_i} x^{s_i}$ where $b_{s_i}\neq 0$ for all $1\leq i\leq w$. We will denote $S=\{s_1,\ldots,s_w\}\subseteq [k']$ the set non-zero coefficient indices of $y$, where we define $[k']=\{0,\ldots,k'-1=k\}$. The coefficients of the product $\alpha^* y=\sum_{j=0}^{k} c_j x^j$ are given by
$$c_j=
\sum_{\substack{a\in A\\s\in S}}\indicator\{a+s\equiv j\pmod{k'}\}b_s.$$
This motivates us to define the shifted set
$$(j-A)_k=\{(j-a_1)\mmod{k'},\ldots,(j-a_d)\mmod{k'}\}$$
which gives us
$$c_j=\sum_{s\in (j-A)_k\cap S}b_{s}.$$
We will denote by
$$J_m=\{j\in [k']\given |(j-A)_k\cap S|=m\}$$
the indices $j$ such that $|(j-A)_k\cap S|$ has size $m$.
It is evident that if $(j-A)_k\cap S=\emptyset$ then $c_j=0$, and if $|(j-A)_k\cap S|=1$ then $c_j\neq 0$. So
$$|J_1|\leq \tweight(\alpha^* y)\leq k-|J_0|.$$
Looking at the conditions of \Cref{lem: equivalent condition}, it would be sufficient to take $c(k)$ such that
$$c(k)-w\leq \min\{|J_0|,|J_1|\}$$
for all $w$.
We make the following claims on lower bounds of $|J_0|$ and $|J_1|$ which will be proven in Section~\ref{sec:proofs-of-claims}.
\begin{claim}
\label{clm: J1 lower bound}
$J_1$ has size at least $wd-2w(w-1)$.
\end{claim}
\begin{claim}
\label{clm: J0 lower bound}
$J_0$ has size at least $k-wd$.
\end{claim}
Assuming the claims, it suffices to take $c(k)$ such that
\begin{align*}
    c(k)\leq& \min_{1\leq w\leq c(k)}wd-2w(w-1)+w
    =\min\{d+1,(d-2c(k)+3)c(k)\},\\
    c(k)\leq& \min_{1\leq w\leq c(k)}k-wd+w
    =k-(d-1)c(k).
\end{align*}
Solving the two inequalities, it suffices to take $c(k)\leq \frac{d}{2}+1$. So $\alpha^*=\sum_{a\in A}x^a$, $c(k)=\frac{d}{2}$ satisfy the condition of \Cref{lem: equivalent condition}.
\end{proof}

\subsection{Proofs of \Cref{clm: J1 lower bound} and \Cref{clm: J0 lower bound}}
\label{sec:proofs-of-claims}
To show a lower bound on $|J_1|$, it is desired that the sets $J_2,..,J_{w}$ are small. The following lemma gives an upper bound on the sizes of $J_2,..,J_{w}$:
\begin{lemma}
\label{lem: J2-Jw bound}
    The sets $J_2,\ldots,J_w$ defined in the proof of \Cref{thm: rate 1/2 construction} satisfy:
    $$\frac{1}{2}\sum_{m=2}^w m|J_m|\leq \sum_{m=2}^w \binom{m}{2}|J_m|\leq 2\binom{w}{2}= w(w-1).$$
\end{lemma}
\begin{proof}
It is evident that the first inequality holds. Let us now show that for any two distinct $s,s'\in S$, we have $\{s,s'\}\subseteq (j-A)_k$ for at most two choices of $j$.
Take any distinct pair $s,s'\in S$. Without loss of generality, assume $s<s'$. Suppose $\{s,s'\}\subseteq (j-A)_k$, then we can write $s\equiv j-a_{l}\pmod{k'}$ and $s'\equiv j-a_{l'}\pmod{k'}$ for some $l,l'\in [d]$. Then $a_l-a_{l'}\equiv s'-s\pmod{k'}$. So 
$$a_l-a_{l'}=\begin{cases}
s'-s & \text{ if }l>l',\\
s'-s-k' & \text{ if } l<l'.
\end{cases}$$
In both cases, by definition of Sidon sets, $l,l'$ are uniquely determined. Thus $j\equiv a_l+s$ is also uniquely determined.

For any $j_m\in J_m$, $|(j_m-A)_k\cap S|=m$ so there are $\binom{m}{2}$ distinct pairs $\{s,s'\}\subseteq (j_m-A)_k$. So the total count of such distinct pairs for all $j\in \cup_{m=2}^w J_m$ is $\sum_{m=2}^ w\binom{m}{2}|J_m|$, which must not exceed $2\binom{w}{2}$ since each pair can only occur in $(j-A)_k$ for two choices of $j$. This gives
$$\sum_{m=2}^w \binom{m}{2}|J_m|\leq 2\binom{w}{2}= w(w-1)$$
which completes the proof.
\end{proof}
We can now obtain a lower bound on the size of $J_1$. 
\begin{C1}
    $J_1$ has size at least $wd-2w(w-1)$.
\end{C1}
\begin{proof}
    Consider the bipartite graph $(S,[k'],E)$ where there is an edge between $s\in S$ and $j\in [k']$ if and only if $j\equiv s+a\pmod{k'}$ for some $a\in A$. Then $J_m$ corresponds to the right vertices $j\in [k']$ with degree $m$. Since each left vertex has degree $|A|=d$, the total number of edges is $|E|=|S|d=wd$. On the other hand, $|E|=\sum_{m=1}^w |J_m|$. Therefore we can write:
    $$|J_1|=|E|-\sum_{m=2}^{w} m|J_m|=wd-\sum_{m=2}^{w} m|J_m|.$$
    The proof is complete via the bound of \Cref{lem: J2-Jw bound}:
    $$|J_1|=wd-\sum_{m=2}^{w} m|J_m|\geq wd-2w(w-1). \qedhere $$
\end{proof}

\begin{C2}
   $J_0$ has size at least $k-wd$.
\end{C2}
\begin{proof}
    Since $s\in (j-A)_k$ is equivalent to $j\in (s+A)_k$, for any $s\in S$, there are $|(s+A)_k|=|A|=d$ values of $j \in [k']$ such that $s\in (j-A)_k$. So there are at most $|S|d=wd$ indices $j$ such that $|(j-A)_k\cap S|>0$. 
\end{proof}

\section{Higher Rate Construction}
\label{section high rate}
In the last section, we gave an explicit construction of rate $1/2$ Wozencraft ensemble code $\CWEs$ achieving minimum distance $\Omega(\sqrt{k})$. In this section, we will show that appropriate puncturing of $\CWEs$ will give us codes of rates $r \in (1/2,1)$ and minimum distance at least $\Omega\left(\left(1-\sqrt{2-\frac{1}{r}}\right)\sqrt{k}\right)$.
\begin{definition}
Let $\CWEr$ be the rate $r$ punctured code given by removing the last $(2-\frac{1}{r})k$ bits of each codeword in $\CWE$.
\end{definition}
\subsection{Analysis of Minimum Distance}
Let $\varphi_r$ denote the map sending any polynomial $f$ to the first $(\frac{1}{r}-1)k$ least significant coefficients. Let $\rweight(f)$ denote the Hamming weight of $\varphi_r(f)$.
\begin{lemma}
    \label{lem: relation of wt_r}
    For any $f\in R$, $\rweight(f\mmod{p})\geq \min\{\rweight(f),(\frac{1}{r}-1)k-\rweight(f)\}$.
\end{lemma}
\begin{proof}
Writing $f=\sum_{i=0}^k b_i x^i$, we have 
$$f\mmod p=f-b_k p=\sum_{i=0}^{k-1}(b_i-b_k)x^i.$$
If $b_k=0$, then $\rweight(f\mmod p)=\rweight(f)$; if $b_k\neq 0$, then 
$$\rweight(f\mmod p)=\Bigl|\bigl\{i\in[(\frac{1}{r}-1)k] \mid b_i\neq b_k\bigr\}\Bigr|\geq (\frac{1}{r}-1) k-\rweight(f).$$
So $\rweight(f\mmod p)\geq \min\{\rweight(f),(\frac{1}{r}-1)k-\rweight(f)\}$.    
\end{proof}
\begin{lemma}
    \label{lem: equivalent condition r}
    Let $\alpha\in \field_{q^k}\subseteq R$. Suppose for every $y\in R$ with $\tweight(y)\leq c(k)$ the condition
    $$c(k)-\tweight(y)\leq \rweight(\alpha y)\leq (\frac{1}{r}-1)k-\left(c(k)-\tweight(y)\right)$$
    holds, where the product $\alpha y$ is taken in $R$. Then $\Delta(\CWEr)\geq c(k)$.
\end{lemma}
\begin{proof}
Take any non-zero $y\in \field_{q^k}\subseteq R$. Its corresponding codeword $\CWEr(y)=(\varphi(y),\varphi_r(\alpha y\mmod{p}))$ has hamming weight $\weight(y)+\weight_r(\alpha y \mmod{p})$. By \Cref{lem: relation of wt_r}, the above condition implies
\begin{align*}
    &\weight(y)+\weight(\alpha y\mmod{p})\\
    \geq& \tweight(y)+\min\{\rweight(\alpha y),(\frac{1}{r}-1)k-\rweight(\alpha y)\}\\
    \geq& c(k).
\end{align*}

Since $\CWEr$ is a linear code, $\Delta(\CWEr)\geq c(k)$.
\end{proof}
The following theorem establishes a lower bound on the minimum distance of the punctured code $\CWErs$ with $\alpha^*$ constructed using Sidon sets as outlined in \Cref{thm: rate 1/2 construction}.
\begin{theorem}
    \label{thm: higher rate construction}
    For any rate $r>1/2$, using the construction of $\alpha^*$ by \Cref{thm: rate 1/2 construction}, the condition of \Cref{lem: equivalent condition r} is satisfied with $c(k)=\Omega\left(\left(1-\sqrt{2-\frac{1}{r}}\right)\sqrt{k}\right)$. Thus the punctured code $\CWErs$ has minimum distance at least $\Omega\left(\left(1-\sqrt{2-\frac{1}{r}}\right)\sqrt{k}\right)$.
\end{theorem}
\begin{proof}
Note that it suffices to find $c(k)=\Omega\left(\left(1-\sqrt{2-\frac{1}{r}}\right)\sqrt{k}\right)$ which satisfies the condition of \Cref{lem: equivalent condition r}. Similar to the proof of \Cref{thm: rate 1/2 construction}, for any $y\in R$ with $\tweight(y)=w$, we write $y=\sum_{i=1}^w b_{s_i} x^{s_i}$ where $b_{s_i}\neq 0$ for all $1\leq i\leq w$ and denote $S=\{s_1,\ldots,s_w\}\subseteq [k']$.
Recall that the coefficients of the product $\alpha y=\sum_{j=0}^{k} c_j x^j$ are given by
$$c_j=\sum_{\substack{a\in A\\s\in S}}\indicator\{a+s\equiv j\pmod{k'}\}b_s=\sum_{s\in (j-A)_k\cap S}b_{s}.$$
We will denote
$$J_m^r=\{j\in [(\frac{1}{r}-1)k]\given |(j-A)_k\cap S|=m\}$$
the non-punctured indices $j$ such that $|(j-A)_k\cap S|$ has size $m$.
Since we have
$$|J_1^r|\leq \rweight(\alpha^* y)\leq (\frac{1}{r}-1)k-|J_0^r| \ ,$$
it would be sufficient to take $c(k)$ such that
$$c(k)-w\leq \min\{|J^r_0|,|J^r_1|\}$$
for all $w$.
We make the following claims on lower bounds of $|J_0^r|$ and $|J_1^r|$, whose proofs we defer to the end of this section.
\begin{claim}
\label{clm: J1r lower bound}
$J_1^r$ has size at least $w\left(d-\sqrt{(2-\frac{1}{r})k}-(2-\frac{1}{r})^{\frac{1}{4}}k^{\frac{1}{4}}\right)-2w^2$.
\end{claim}
\begin{claim}
\label{clm: J0r lower bound}
$J_0^r$ has size at least $(\frac{1}{r}-1)k-w\left(\sqrt{(\frac{1}{r}-1)k}+(\frac{1}{r}-1)^\frac{1}{4}k^\frac{1}{4}+1\right)$.
\end{claim}
Assuming the claims, it suffices to take $c(k)$ such that
\begin{align*}
    c(k) & \leq m_1:=\min_{1\leq w\leq c(k)}w\left(d-\sqrt{(2-\frac{1}{r})k}-(2-\frac{1}{r})^{\frac{1}{4}}k^{\frac{1}{4}}+1\right)-2w^2,\\
c(k) & \leq m_2:=\min_{1\leq w\leq c(k)}(\frac{1}{r}-1)k-w\left(\sqrt{(\frac{1}{r}-1)k}+(\frac{1}{r}-1)^\frac{1}{4}k^\frac{1}{4}\right).
\end{align*}
Note that $m_1$ is the minimum of a quadratic polynomial in $w$ which is obtained either when $w=1$ or $w=c(k)$; $m_2$ is the minimum of a linear function in $w$ which is obtained when $w=c(k)$. So
\begin{align*}
    m_1=&\min\bigg\{d-\sqrt{(2-\frac{1}{r})k}-(2-\frac{1}{r})^{\frac{1}{4}}k^{\frac{1}{4}},\ 
    c(k)\left(-2c(k)+d-\sqrt{(2-\frac{1}{r})k}-(2-\frac{1}{r})^{\frac{1}{4}}k^{\frac{1}{4}}+1\right)\bigg\},\\
    m_2=&(\frac{1}{r}-1)k-c(k)\left(\sqrt{(\frac{1}{r}-1)k}+(\frac{1}{r}-1)^\frac{1}{4}k^\frac{1}{4}\right).
\end{align*}
We want $c(k)\leq m_1$ and $c(k)\leq m_2$. Therefore, we have the following restrictions on $c(k)$:
\begin{align*}
    c(k)\leq& \frac{1}{2}\left(d-\sqrt{(2-\frac{1}{r})k}-(2-\frac{1}{r})^{\frac{1}{4}}k^{\frac{1}{4}}\right)
    =\frac{1}{2}\left(1-\sqrt{2-\frac{1}{r}}-o(1)\right)\sqrt{k},\\
    c(k)\leq& \frac{(\frac{1}{r}-1)k}{\sqrt{(\frac{1}{r}-1)k}+(\frac{1}{r}-1)^\frac{1}{4}k^\frac{1}{4}+1}
    =(1-o(1))\sqrt{(\frac{1}{r}-1)k}.
\end{align*}
Recall that $d$ is the largest prime smaller than $\sqrt{k}$ and $d=(1-o(1))\sqrt{k}$ by \Cref{remark d}. Since $\frac{1}{2}(1-\sqrt{2-\frac{1}{r}})<\sqrt{(\frac{1}{r}-1)}$ for $r\in (\frac{1}{2},1)$, we can take $c(k)=\frac{1}{2}(d-\sqrt{(2-\frac{1}{r})k}-(2-\frac{1}{r})^{\frac{1}{4}}k^{\frac{1}{4}})$ and thus $\CWErs$ has minimum distance at least $c(k)=\Omega\left(\left(1-\sqrt{2-\frac{1}{r}}\right)\sqrt{k}\right)$.
\end{proof}
\subsection{Proofs of \Cref{clm: J1r lower bound} and \Cref{clm: J0r lower bound}}
To prove the two claims, we will first introduce a theorem by Erd\H{o}s and Tur\'{a}n \cite{erdos1941problem} which bounds the order of any Sidon set via its length.
\begin{theorem}[Erd\H{o}s and Tur\'{a}n \cite{erdos1941problem, lindstrom1972b2, obryant2022sizefinitesidonsets}]
\label{thm: erdos}
    For any Sidon set with length $m$, the order is at most $\sqrt{m}+m^{1/4}+1$.
\end{theorem}
A key observation is that for any $s\in S$, the shifted set $(s+A)_k$ is a Sidon set. Moreover, it is partitioned into two parts after puncturing: the remaining part $(s+A)_k\cap [(\frac{1}{r}-1)k]$ and the removed part  $(s+A)_k\cap ([k']\setminus [(\frac{1}{r}-1)k])$, where each part is itself a Sidon set. We can then apply \Cref{thm: erdos} to the two parts. In general, for any $m\leq k$, the following lemma bounds the number of index $j<m$ such that $s\in (j-A)_k$.
\begin{lemma}
    \label{lem: j solutions}
    For any $s\in [k']$ and $m\leq k$,
    $$d-(\sqrt{k-m}+(k-m)^{1/4}+2)\leq |[m]\cap (s+A)_k|\leq \sqrt{m}+m^{1/4}+1$$
\end{lemma}
\begin{proof}
Translating the sets by $s\mmod{k}$, we have 
$$|[m]\cap (s+A)_k|=|([m]-s)_k\cap A|.$$
Observe that $A\subseteq [d^2-1]$ and $B:=([m]-s)_k\cap [d^2-1]$ are sets of consecutive numbers modulo $d^2-1$ with size at most $m$. Since $A$ is a Sidon set modulo $d^2-1$, $A\cap B$ is also a Sidon set modulo $d^2-1$ of length at most $m$. By \Cref{thm: erdos}, the order of $A\cap B$ is at most $\sqrt{m}+m^{1/4}+1$. So
$$|[m]\cap (s+A)_k|=|A\cap B|\leq \sqrt{m}+m^{1/4}+1.$$
On the other hand, let $C:=([k']\setminus[m]-s)_k\cap [d^2-1]$. Then $C$ is a set of consecutive numbers modulo $d^2-1$ with size at most $k'-m$. Similarly, we have $A\cap C$ a Sidon set modulo $d^2-1$ of length at most $k'-m=k-m+1$ and order at most $\sqrt{k-m}+(k-m)^{1/4}+2$. Since $|A\cap B|+|A\cap C|=d$, it follows that
$$d-(\sqrt{k-m}+(k-m)^{1/4}+2)\leq |[m]\cap (s+A)_k|\leq \sqrt{m}+m^{1/4}+1.\qedhere$$
\end{proof}
We can now prove \Cref{clm: J1r lower bound} and \Cref{clm: J0r lower bound}:
\begin{C3}
    $J_1^r$ has size at least $w\left(d-\sqrt{(2-\frac{1}{r})k}-(2-\frac{1}{r})^{\frac{1}{4}}k^{\frac{1}{4}}\right)-2w^2$.
\end{C3}
\begin{proof}
Consider the bipartite graph $(S,[(\frac{1}{r}-1)k],E)$ where there is an edge between $s\in S$ and $j\in [(\frac{1}{r}-1)k]$ if and only if $j\equiv s+a\pmod{k'}$ for some $a\in A$. Then $J_m^r$ corresponds to the right vertices $j\in [(\frac{1}{r}-1)k]$ with degree $m$. Apply \Cref{lem: j solutions} with $m=(\frac{1}{r}-1)k$, each left vertex has degree at least $d-\sqrt{(2-\frac{1}{r})k}-(2-\frac{1}{r})^{\frac{1}{4}}k^\frac{1}{4}-2$. The total number of edges $|E|$ is at least $w\left(d-\sqrt{(2-\frac{1}{r})k}-(2-\frac{1}{r})^{\frac{1}{4}}k^\frac{1}{4}-2\right)$. On the other hand, $|E|=\sum_{m=1}^w |J_m^r|$. Therefore we can write:
    $$|J_1^r|=|E|-\sum_{m=2}^{w} m|J_m^r|\geq w(d-\sqrt{(2-\frac{1}{r})k}-(2-\frac{1}{r})^{\frac{1}{4}}k^\frac{1}{4}-2)-\sum_{m=2}^{w} m|J_m^r|.$$
    By \Cref{lem: J2-Jw bound},
    $$\sum_{m=2}^{w} m|J_m^r|\leq \sum_{m=2}^{w} m|J_m|\leq 2w(w-1),$$
    which completes the proof.
\end{proof}
\begin{C4}
    $J_0^r$ has size at least $(\frac{1}{r}-1)k-w\left(\sqrt{(\frac{1}{r}-1)k}+(\frac{1}{r}-1)^\frac{1}{4}k^\frac{1}{4}+1\right)$.
\end{C4}
\begin{proof}
Apply \Cref{lem: j solutions} with $m=(\frac{1}{r}-1)k$, for each $s\in S$, there at most $\sqrt{(\frac{1}{r}-1)k}+(\frac{1}{r}-1)^{\frac{1}{4}}k^\frac{1}{4}+1$ many $j\in [(\frac{1}{r}-1)k]$ such that $s\in (j-A)_k$. So overall there are at most $|S|(\sqrt{(\frac{1}{r}-1)k}+(\frac{1}{r}-1)^{\frac{1}{4}}k^\frac{1}{4}+1)$ many $j\in [(\frac{1}{r}-1)k]$ such that $|(j-A)_k\cap S|>0$, which gives the claimed lower bound on $|J_0^r|$.
\end{proof}

\section{Open Questions}
\label{sec:open}
It is well known that Wozencraft ensemble codes $\CWE$ satisfy the Gilbert-Varshamov bound for most $\alpha$. Concretely,
$$
    \lim_{k\rightarrow \infty}\Pr_{\alpha\in S}\left[\Delta(\CWE)\geq d(k)\right]=1,
$$
where $S=\field_{q^k}^*,\ d(k)=\left(h_q^{-1}(\frac{1}{2})-\epsilon\right)\cdot 2k
$ with $h_q$ the $q$-ary entropy function and $\epsilon>0$ chosen arbitrarily. In this paper, we have proposed a construction of $S=\{\alpha^*\}$ such that the equation above holds with $d(k)=\Omega(\sqrt{k})$. It is of natural interest to reduce the size of the ensemble and find $S,d(k)$ satisfying the equation above with $|S|$ small and $d(k)$ large. For example, can one construct $S$ such that the equation above is satisfied via $|S|=O(2^{o(k)})$ and $d(k)=\Omega(k)$, or $|S|=\text{poly}(k)$ and $d(k)=\Omega(k^c)$ with $c>\frac{1}{2}$? In addition, are there any barriers to such constructions, as in, would such constructions imply progress on some other explicit construction challenge?

In 1973, Weldon \cite{weldon1973lowrate} proposed an ensemble of codes that generalizes Wozencraft ensemble codes.  The Weldon ensemble  was used by him, and later Shen~\cite{shen1993justesen}, to construct explicit concatenated codes achieving the Zyablov bound for rates less than $0.3$, thus improving upon Justesen codes \cite{justesen1972class} for low rates. Weldon codes\footnote{The ensemble codes which Weldon \cite{weldon1973lowrate} proposed can have any rate of $n_1/(n_1+n_2)$ with $n_1,n_2$ positive integers. However, he only used codes of rate $1/(t+1)$ to construct the concatenated code.} of rate $1/(t+1)$ are indexed by $\alpha_1,\alpha_2,\ldots,\alpha_t\in\field_{q^k}$ and defined as
$$C_{\text{WE}}^{\alpha_1,\ldots,\alpha_t}=(\varphi(x),\varphi(\alpha_1 x),\ldots,\varphi(\alpha_t x)).$$
For some fixed $t > 2$, can one find explicit $\alpha_i$'s where the distance of $C_{\text{WE}}^{\alpha_1,\ldots,\alpha_t}$ is asymptotically larger than $\sqrt{k}$?

\bibliography{main}
\bibliographystyle{alpha}
\end{document}